\documentclass[10pt,conference]{IEEEtran}

\usepackage{graphics,
           psfrag,
           epsfig,
           amsthm,
           cite,
           amssymb,
           url,
           dsfont,
           subfigure,
           algorithm,
           algorithmic,
           balance,
           enumerate,
           color,
           setspace,
           multirow,
           mathtools
}
\usepackage{amsmath}

\newtheorem{example}{Example}
\newtheorem{definition}{Definition}

\newtheorem{lemma}{Lemma}

\newtheorem{theorem}{Theorem}
\newtheorem{corollary}{Corollary}

\newtheorem{hypothesis}{Hypothesis}

\newcommand{\sref}[1]{Section~\ref{#1}}
\newcommand{\appref}[1]{Appendix~\ref{#1}}

\newcommand{\dref}[1]{Definition~\ref{#1}}

\newcommand{\cref}[1]{Corollary~\ref{#1}}
\newcommand{\thref}[1]{Theorem~\ref{#1}}
\newcommand{\lref}[1]{Lemma~\ref{#1}}

\newcommand{\ignore}[1]{}

\begin{document}

\title{A Graph Model for Opportunistic Network Coding}
\author{\IEEEauthorblockN{Sameh Sorour${^*}$, Neda Aboutorab$^\dagger$, Parastoo Sadeghi$^\dagger$, Tareq Y. Al-Naffouri${^{*\ddagger}}$, and Mohamed-Slim Alouini${^\ddagger}$}
\IEEEauthorblockA{${^*}$Electrical Engineering Department, King Fahd University of Petroleum and Minerals (KFUPM), Dhahran, Saudi Arabia\\
$^\dagger$Research School of Engineering, The Australian National University, Canberra, Australia\\$^\ddagger$CEMSE Division, King Abdullah University of Science and Technology (KAUST), Thuwal, Saudi Arabia\\
$^*$samehsorour@kfupm.edu.sa,$^\dagger$\{neda.aboutorab,parastoo.sadeghi\}@anu.edu.au,$^\ddagger$\{tareq.alnaffouri,slim.alouini\}@kaust.edu.sa}}

\maketitle

\begin{abstract}
Recent advancements in graph-based analysis and solutions of instantly decodable network coding (IDNC) trigger the interest to extend them to more complicated opportunistic network coding (ONC) scenarios, with limited increase in complexity.\ignore{ This requires the design of a simple IDNC-like graph model for these scenarios, which captures their properties and capacities.} In this paper, we design a simple IDNC-like graph model for a specific subclass of ONC, by introducing a more generalized definition of its vertices and the notion of \emph{vertex aggregation} in order to represent the storage of non-instantly-decodable packets in ONC. Based on this representation, we determine the set of pairwise vertex adjacency conditions that can populate this graph with edges so as to guarantee decodability or \emph{aggregation} for the vertices of each clique in this graph. We then develop the algorithmic procedures that can be applied on the designed graph model to optimize any performance metric for this ONC subclass. A case study on reducing the completion time shows that the proposed framework improves on the performance of IDNC and gets very close to the optimal performance.
\end{abstract}

\section{Introduction}\label{sec:intro}
Opportunistic network coding (ONC) was recently shown to provide fast and tailored packet delivery according to real-time demands of receivers and their side information \cite{4476183,4594999,Drinea2009,Gatzianas2010,Wang2010}\ignore{, as compared to block-based random linear network coding (RLNC)}. The philosophy behind ONC is to exploit the previously stored source packets and \emph{undecoded packet combinations} at the receivers in selecting subsequent combinations to maximize the network gains. By an undecoded packet combination for a given receiver, we mean a coded combination of source packets including more than one missing source packet at that receiver, thus preventing it from decoding a new source packet when received. Despite the benefits in storing such combinations for future decoding instances, exploring these benefits makes the proposed ONC algorithms in the literature highly complex and unscalable, as they need to examine an exploding number of sets and virtual queues to find suitable network codes. Moreover,\ignore{ unlike graph-based IDNC \cite{6030131},} there does not exist a general ONC framework to optimize any desired metric in a network.

On the other hand, many advancements have been achieved in analyzing, optimizing, and designing simple algorithms for a special subclass of ONC, namely instantly decodable network coding (IDNC), which does not allow any receiver to store undecoded packet combinations. These advancement were made possible thanks to the \emph{IDNC's simple graph model} \cite{TON10-CD,GC10,Li2011,6030131,ISIT12}. Despite the promising performance and easy scalability of the proposed IDNC algorithms\ignore{ and their easy scalability to large networks}, the fact of not storing and exploiting undecoded packet combinations in future decoding instances remains a clear limitation against exploring further performance improvements.

The above facts motivate the possibility of\ignore{ extending the recent IDNC advancements to} simplifying\ignore{ the analysis and} the algorithm design for general ONC, by creating a simple \emph{IDNC-like graph model} for it. Indeed, if such model is found, we can extend the IDNC graph-based analysis to ONC, and thus design\ignore{ similar} graph-based ONC algorithms with no or limited increase in complexity. For this ONC graph to be IDNC-like, it should be simple to construct using only \emph{pairwise vertex adjacency conditions (P-VACs)} that generate graph edges only based on pairwise vertex relations (as opposed to examining the sets and virtual queues of all the receivers, which is the complicating factor in current ONC algorithms). Yet, it should both capture the new properties that result from storing undecoded packet combinations and represent the possibility of utilizing them in future decoding instances.

In this paper, we focus on a subclass of ONC, which we will refer to as \emph{Order 2 ONC (O2-ONC)}. Unlike IDNC, receivers in O2-ONC are allowed to store and utilize \textit{Order-2 innovative packets (O2-IPs)}, which are initially defined for any given receiver as the coded packets including only two missing source packets\ignore{ that were not previously decoded} at this receiver (We will revisit this definition in Definition \ref{def:O2-IP}). But even with this one step extension from IDNC, the graph representation of O2-ONC is not trivial. In particular, one needs to properly cater for both the generation of vertices representing stored O2-IPs, and the vertex manipulations when further O2-IPs are stored. Proper P-VACs must also be derived between these newly defined vertices, such that the graph cliques represent valid packet combinations providing instantaneous benefits to their vertices. 

One contribution of this paper is the definition of a new vertex representation for the O2-ONC graph, the novel concept of \emph{vertex aggregation}, and the set of P-VACs that can populate this graph with edges, so as to guarantee decodability or \emph{aggregation} for all the vertices of each clique. Another main contribution is developing algorithmic procedures that can be applied on the designed graph model to optimize any O2-ONC performance metric. We finally present a case study on solving the completion time problem for O2-ONC using our proposed framework and compare the results to both IDNC and the optimal performance over all linear network codes.

We conclude this section by noting that our earlier work in \cite{ICC13-Neda} could only represent some of the information that stored O2-IPs bring, by adding extra vertices to the conventional IDNC graph, whereas this work re-builds the entire graph to capture all information about any stored O2-IP at the receivers.

\section{System Model}
We consider a wireless sender transmitting a frame $\mathcal{N}$ of source packets to a set $\mathcal{M}$ of receivers.\ignore{ over heterogenous erasure channels.Each receiver is interested in receiving a subset or all the packets of $\mathcal{B}$. The sender first transmits the packets in $\mathcal{N}$ without coding.} We assume that the receivers initially hold different (possibly overlapping) subsets of these packets and the sender aims to deliver the rest of these packets to them.  At any snapshot of time during this \emph{delivery process}, two sets of packets are attributed to each receiver $i$:
\begin{itemize}
\item The \textit{Has} set ($\mathcal{H}_i$) is defined as the set of source packets (that is, excluding O2-IPs) received by receiver $i$.
\item The \textit{Wants} set ($\mathcal{W}_i$) is defined as $\mathcal{N} \setminus \mathcal{H}_i$\ignore{ is defined as the set of source packets that are not received by $i$. This include the XORed packets in any stored O2-IP combination at $i$}.
\end{itemize}
Due to the diversity in this side information at different receivers, the sender employs ONC to deliver the missing packets. At each transmission, the sender must make a decision on which source packets to combine and send to benefit a certain set of receivers, given the O2-ONC constraint. In other words, the sender must select these coded combinations with the knowledge that any receiver will discard a packet if it is neither instantly decodable nor O2-innovative, which makes such combination of no benefit to this receiver.

In the rest of the paper, we will use $r$, $p$ and $\mathcal{P}$ to denote receivers, source packets and combinations of source packets, respectively. Also, a packet combination $\mathcal{P}$ can be interchangeably interpreted, according to the context, as either the set of source packets in this packet combination or the coded packet resulting from combining these source packets. Moreover, the sender is assumed to use either higher Galois field operations or triangular XORs (interested readers are refereed to \cite{6275780}) in the encoding processes, such that any receiver $r_i$ having $n$ O2-IPs $\mathcal{P}_1, \dots, \mathcal{P}_n$ can decode a set $x = \bigcup_{k=1}^n \left(\mathcal{P}_{k} \cap \mathcal{W}_i\right)$ of source packets if $|x| = n$. In other words, all received and stored O2-IPs are all assumed to be linearly independent from each other. This is a common assumption in most works on ONC algorithm design and can be guaranteed almost surely by employing an appropriate Galois field size or appropriate overhead in triangular XORs at the sender.

\section{Background on IDNC Graph}\label{sec:IDNC}
The IDNC graph defines the set of all feasible instantly decodable packet combinations for the IDNC paradigm and determines the instantly decoding receivers of each of them. This graph $\mathcal{G}(\mathcal{V},\mathcal{E})$ is constructed by first inducing a vertex $v_{i,j}$ in $\mathcal{V}$ for each $p_j \in \mathcal{W}_i$, $\forall~r_i\in\mathcal{M}$. In other words, any vertex $v_{i,j}$ represents a wanted $p_j$ by $r_i$.\ignore{ We call the $i$ and $j$ indices of any vertex $v_{i,j}$ as its receiver and packet elements, respectively.}

The P-VACs in the IDNC graph are set as follows. Two vertices $v_{i,j}$ and $v_{k,l}$ in $\mathcal{V}$ are connected by an edge in $\mathcal{E}$ if any one of the following conditions is true:
\begin{itemize}
\item C1: $p_j = p_l$ $\Rightarrow$ The two vertices represent the need of $p_j$ by two different receivers $r_i$ and $r_k$.
\item C2: $p_j\in \mathcal{H}_k$ and $p_l \in \mathcal{H}_i$ $\Rightarrow$ The wanted packet of each vertex is in the Has set of the receiver\ignore{ that induced} of the other vertex.
\end{itemize}
It can be easily shown that the set of P-VACs are necessary and sufficient conditions for the possibility of serving all the demands of the vertices of any clique in the graph by one transmission, consisting of either one source packet or an XOR of several source packets that are identified by the clique vertices \cite{TON10-CD,GC10}.
It can also be easily shown that the construction of this graph needs $O(M^2N)$ operations, such that $M=|\mathcal{M}|$ and $N=|\mathcal{N}|$ \cite{TON10-CD}.

\section{O2-ONC Graph: Vertices} \label{sec:vertices}
Similar to the IDNC graph, each vertex of a given receiver in the O2-ONC graph should represent the need to receive a certain packet that benefits this receiver.\ignore{ Edges in the O2-ONC graph between any group of vertices forming a clique should point to the existence of a packet combination that can simultaneously serve the needs of all these vertices, thus simultaneously providing benefits to all the receivers inducing these vertices.} Since IDNC  does not allow the storage of undecoded packet combinations, the only benefit that the sender can achieve by serving $v_{i,j}$ is the delivery of $p_j$ to $r_i$, which makes $r_i$ one step closer to completing the reception of all its required source packets. By removing $v_{i,j}$ after its service, the number of vertices belonging to $r_i$ in the IDNC graph represents the number of needed packets until its service completion.

Now, in addition to the above notion, storing O2-IPs in O2-ONC extends the extent of receiver benefits. Indeed, when $r_i$, having $p_j$ and $p_l$ in $\mathcal{W}_i$, receives and stores a packet $p_j\oplus p_l$ at the $n$-th transmission, it requires the reception of one of these packets or their combination (either with other GF(4) coefficients or with triangular XOR) to decode both packets. So, after the reception of this packet combination $p_j\oplus p_l$ in the $n$-th transmission, $r_i$ becomes one step closer to its service completion. Indeed, if $r_i$ had to receive $|\mathcal{W}_i|$ packets to reach completion before the $n$-th transmission, it now needs to benefit from only $|\mathcal{W}_i|-1$ transmissions to reach completion.

This instance and this change of receiver requirements must be reflected in the O2-ONC graph. The packets $p_j$ and $p_l$ should not be represented by two vertices $v_{i,j}$ and $v_{i,l}$ any longer because they can be both decoded by one transmission. Thus, we define a new vertex representation $v_{i,j\cup l}$ that can represent this fact by only one vertex in the O2-ONC graph. We will refer to this notion as \emph{vertex aggregation} as it can be seen as the aggregation of the two vertices $v_{i,j}$ and $v_{i,l}$ into one bigger vertex. This description motivates the following generalization of the vertex definition in the O2-ONC graph.

\begin{definition}[Generalized Definition of Vertices] \label{def:generalized-vertex}
A vertex $v_{i,x}$ in the O2-ONC graph represents both:
\begin{itemize}
\item A set $x$ of packets, $|x|\geq 1$, desired by $r_i$.
\item A set of $|x|-1$ stored O2-IPs $\mathcal{P}_1, \dots, \mathcal{P}_{|x|-1}$ at $r_i$, such that all the following is true $\forall~k\in\{1, \dots, |x|-1\}$:
\begin{itemize}
\item $\mathcal{P}_k \subseteq x\cup\mathcal{H}_i$.
\item $\mathcal{P}_k\cap x \neq \emptyset$.
\item $\bigcup_{k=1}^n \left(\mathcal{P}_{k} \cap \mathcal{W}_i\right) = x$
\end{itemize}
In other words, $r_i$ has $|x|-1$ O2-IPs, each of which including a subset or all the source packets of $x$ and strictly NO source packets in $\mathcal{W}_i\setminus x$. Moreover, these packets must collectively include all the elements of $x$.
\end{itemize}
\end{definition}
We will call the $x$ index of vertex $v_{i,x}$ as the packet set of this vertex. We further define the dimension of a vertex $v_{i,x}$ as the cardinality of its packet set $x$ (i.e. $|x|$). It is important to mention here that \dref{def:generalized-vertex} is introduced only for theoretical explanations and will not be practically used to generate vertices in the O2-ONC graph, as this will result in a large complexity. As will be explained in \sref{sec:construction}, the O2-ONC graph luckily starts as an IDNC graph as no receiver stores any O2-IPs in the beginning of the delivery phase. Vertex aggregation will thus occur progressively after each transmission if needed using a much simpler procedure.

According to the above generalized definition of vertices in the O2-ONC graph, all the packets in the set $x$ of each vertex $v_{i,x}$ can be decoded by only one packet combination $\mathcal{P}$ when both following decodability conditions hold:
\begin{itemize}
\item $\mathcal{P}\subseteq x\cup\mathcal{H}_i$.
\item $\mathcal{P}\cap x \neq \emptyset$.
\end{itemize}
Indeed, when such packet is received by $r_i$, it can cancel the packets in $\mathcal{H}_i$, such that the remaining combination has packets that are all in $x$. Thus, $r_i$ will possess now $|x|$ independent linear equations in $|x|$ variables, which can be used to decode all the packets in $x$. Also note that the above definition includes the IDNC-type vertices when $|x|=1$ (which represents $|x|-1 = 0$ stored O2-IPs). The above concepts about generalized vertices are illustrated by Example \ref{ex:generalized-vertex} in \appref{app:examples}.

\ignore{
\begin{example}\label{ex:generalized-vertex}
A receiver $r_i$ with $\mathcal{H}_i = \{p_1,p_3,p_5,p_7\}$ and $\mathcal{W}_i = \{p_2,p_4,p_6,p_8\}$ is storing the following combinations:
\begin{itemize}
\item $\mathcal{P}_a = a_1 p_1 + a_2 p_2 + a_3 p_4$.
\item $\mathcal{P}_b = b_1 p_3 + b_2 p_4 + b_3 p_7 + b_4 p_8$.
\end{itemize}
Clearly, these 2 O2-IPs collectively include only 3 source packets from $\mathcal{W}_i$ and satisfy the conditions in \dref{def:generalized-vertex}. Consequently, these 3 source packets can be represented in the O2-ONC graph by only ONE vertex $v_{i,\{2,4,8\}}$. This representation means that $r_i$ can decode all three source packets if it receives only ONE of these source packets or ONE appropriate combination of any of  them that satisfy the above decodability conditions. For example, if $r_i$ receives $p_2$, $p_4$, $p_8$, or any combination of two or three or them (and possibly combined with other packets in $\mathcal{H}_i$, it will have 3 equations in 3 unknown source packets, which will result in their decoding.
\end{example}
}

The following lemma extends the above notion of vertex aggregation to this new generalized definition of vertices.
\begin{lemma}\label{lem:aggregation}
For any two vertices $v_{i,x}$ and $v_{i,y}$, $|x|,|y|\geq 1$, if receiver $r_i$ receives a packet combination $\mathcal{P}$, such that
\begin{itemize}
\item $\mathcal{P} \subseteq\{x\cup y \cup \mathcal{H}_i\}$
\item $\mathcal{P}\cap x \neq \emptyset$ and $\mathcal{P}\cap y \neq \emptyset$,
\ignore{\item $\mathcal{P} \subseteq\{y\cup \mathcal{H}_k\}$ and $\mathcal{P}\cap y \neq \emptyset$,}
\end{itemize}
then, these two vertices can be aggregated into one single vertex $v_{i,x\cup y}$.
\end{lemma}

\begin{proof}
The proof can be found in \appref{app:aggregation}
\end{proof}

The proof of \lref{lem:aggregation} sheds light onto an extended definition for the O2-IPs that matches the new vertex definition.

\begin{definition}[O2-IPs and Aggregating Packets] \label{def:O2-IP}
A packet combination $\mathcal{P}$ is said to be an O2-IP (or an aggregating packet) for $r_i$ if and only if $r_i$ has two vertices $v_{i,x}$ and $v_{i,y}$ such that:
\begin{itemize}
\item $\mathcal{P} \subseteq\{x\cup y\cup \mathcal{H}_i\}$
\item $\mathcal{P}\cap x \neq \emptyset$ and $\mathcal{P}\cap y \neq \emptyset$
\end{itemize}
In other words, the elements of $\mathcal{P}$ that are not in $\mathcal{H}_i$ must belong to the packet sets of only 2 vertices of receiver $r_i$.
\end{definition}
Note that we called these packets defined in \dref{def:O2-IP} as aggregating packets because when $r_i$ receives such packets, both vertices $v_{i,x}$ and $v_{i,y}$ will be aggregated into only one vertex $v_{i,x\cup y}$ as mandated by \lref{lem:aggregation}. This concept is illustrated by Example \ref{ex:aggregating-packets} in \appref{app:examples}.

\ignore{
\begin{example} \label{ex:aggregating-packets}
In the same scenario of Example \ref{ex:generalized-vertex}, if $r_i$ receives a coded packet $\mathcal{P}_c = c_1 p_2 + c_2 p_5 + c_3 p_6 + c_4 p_8$, it can be seen that, by excluding the Has set packets of $r_i$ from $\mathcal{P}_c$, the remaining source packets belong to either vertex $v_{i,\{2,4,8\}}$ or $v_{i,6}$. Despite the fact that $\mathcal{P}_c$ has more than 2 (in fact 3) packets from $\mathcal{W}_i$, it still is considered as an $O2-IP$ according to \dref{def:O2-IP} and its reception at $r_i$ results in aggregating the vertices $v_{i,\{2,4,8\}}$ and $v_{i,6}$ into one vertex $v_{i,\{2,4,6,8\}}$. Indeed, $r_i$ now needs to receive only ONE of these source packets or ONE combination of them (again possibly with other packets from $\mathcal{H}_i$), which can be used along with $\mathcal{P}_a$, $\mathcal{P}_b$ and $\mathcal{P}_c$ to decode all of these four source packets. Thus, they must be all represented by ONE vertex in the O2-ONC graph.
\end{example}
}


\section{O2-ONC Graph: Edges and P-VACs} \label{sec:edges}
As in the IDNC graph, an edge between any two vertices in the O2-ONC graph should reflect a possibility of simultaneous benefit for their receivers. More generally, for any group of vertices forming a clique, there must exist a packet combination that can simultaneously benefit all the receivers inducing these vertices. We will refer to this property as the clique benefit property and to the cliques that satisfy this property as \emph{proper cliques}.

\begin{definition}[Proper Cliques and Proper $K$-Cliques]\label{def:proper-cliques}
A proper clique in the O2-ONC graph is a clique for which there exists at least one packet combination that can either decode (hence remove) or further aggregate each and every vertex of that clique. A proper $K$-clique is a proper clique of size $K$ (i.e. a proper clique that consists of $K$ vertices).
\end{definition}
\ignore{
\begin{figure}[t]
\centering
\includegraphics[width=0.6\linewidth]{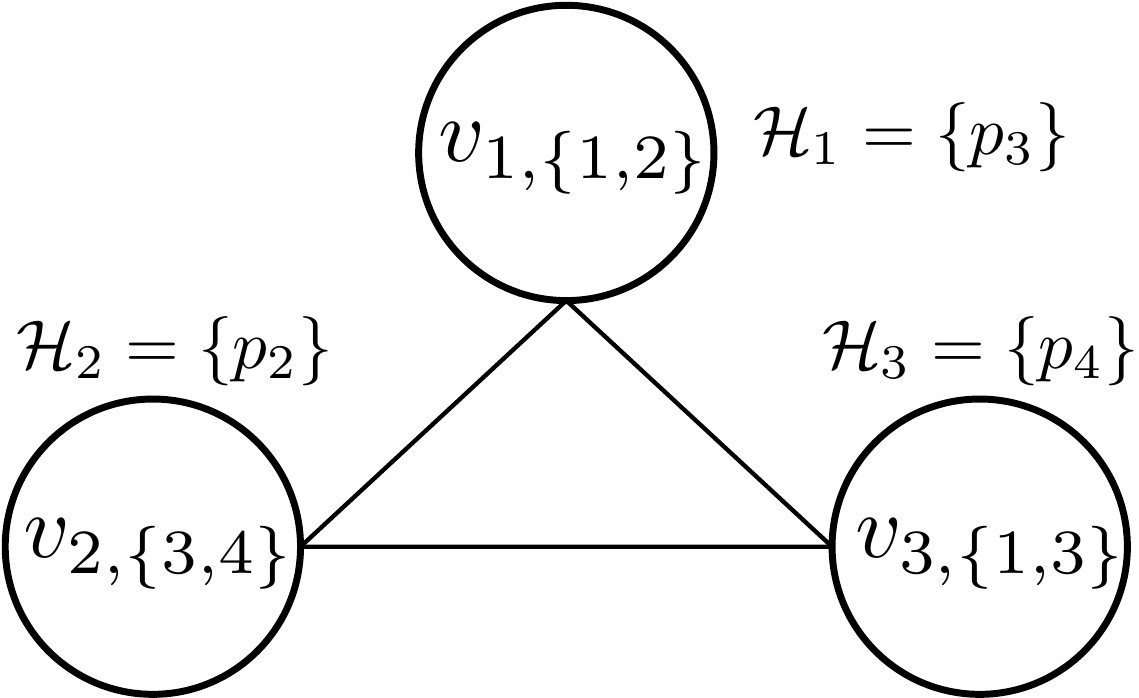}\\
\caption{Example of a proper clique between 3 vertices in the O2-ONC graph (The rest of the vertices are hidden to emphasize this clique). Each two vertices satisfy either of the Simple P-VACs in \dref{def:Simple-PVACs}. A coding combination $\{p_1,p_3\}$ or $\{p_1,p_4\}$ or $\{p_2,p_3\}$ or $\{p_2,p_4\}$ will definitely either decode each of these vertices or will aggregate it with another vertex of the same receiver. For example, $\{p_1,p_3\}$ will decode vertices $v_{1,\{1,2\}}$ and  $v_{3,\{1,3\}}$, and will aggregate $v_{2,\{3,4\}}$ and $v_{2,1}$ into one vertex $v_{2,\{1,3,4\}}$.}
\label{fig:proper-clique}
\end{figure}
Fig. \ref{fig:proper-clique} shows an example of a proper 3-clique in the shown scenario. We can easily see that all combinations $\{p_1,p_3\}$ or $\{p_1,p_4\}$ or $\{p_2,p_3\}$ or $\{p_2,p_4\}$ will definitely either decode each of these vertices or will aggregate it with another vertex of the same receiver.
}
Example \ref{ex:proper-clique} in \appref{app:examples} depicts an illustration of a proper 3-clique. As mentioned in \sref{sec:intro}, edges should be generated through P-VACs, which consider only pairwise relations between any two vertices it may connect. Similar to IDNC, this property is important to ensure that we do not significantly increase the graph construction complexity. At the same time, these P-VACs must be designed such that they can guarantee that all generated cliques in the O2-ONC graph are proper.

\subsection{Simple P-VACs}
By first looking at the edge level (i.e. 2-cliques), it is not difficult to show that they can be made proper if the following set of Simple P-VACs are defined as follows:
\begin{definition}[Simple P-VACs]\label{def:Simple-PVACs}
\ignore{ The Simple P-VACs is defined by the following two adjacency conditions:}
Two vertices $v_{i,x}$ and $v_{k,y}$ are adjacent in the O2-ONC graph according to one of the two following benefiting conditions (BC):
\begin{itemize}
\item BC1: If $x\cap y \neq \emptyset ~ \Rightarrow ~ v_{i,x}$ is set adjacent to $v_{k,y}$ with no other conditions.
\item BC2: If $x\cap y = \emptyset ~ \Rightarrow ~ v_{i,x}$ is set adjacent to $v_{k,y}$ only if $x\cap \mathcal{H}_k \neq \emptyset$ AND $y\cap \mathcal{H}_i \neq \emptyset$
\end{itemize}
\end{definition}
\ignore{Every two vertices in the example of Fig. \ref{fig:proper-clique} satisfy one of the above simple P-VACs. Indeed, BC1 is satisfied between vertices $v_{1,\{1,2\}}$ and $v_{3,\{1,3\}}$ as well as between $v_{2,\{3,4\}}$ and $v_{3,\{1,3\}}$. BC2 is also satisfied between $v_{1,\{1,2\}}$ and $v_{2,\{3,4\}}$. Consequently,} Example \ref{ex:proper-clique} in \appref{app:examples} shows that the above conditions can generate proper 2-cliques and proper 3-cliques. However, the following theorem shows the limitations on these P-VACs in generating proper $K$-cliques for $K>3$.

\begin{theorem}\label{th:BC-bound}
The Simple P-VACs guarantee the formation of $K$-proper cliques for $K\leq 3$. Moreover, considering Condition BC1 only, proper K-cliques can be guaranteed for $K\leq 4$.
\end{theorem}
\begin{proof}
The proof can be found in \appref{app:BC-bound}
\end{proof}

\ignore{
\begin{proof}
\thref{th:simple-PVACs} proves the first statement of the theorem for $K \leq 3$.\\
For $K=4$, assume four vertices $v_{i,x}$, $v_{j,y}$, $v_{k,z}$ and $v_{m,w}$ forming a 4-clique because they pairwise satisfy Condition C1. For this clique, we can always generate the following packet combination:
\begin{equation}
\mathcal{P} = p_{x\cap y} \oplus p_{z\cap w}
\end{equation}
Clearly, the first source packet $p_{x\cap y}$ is instantly decodable for vertices $v_{i,x}$ and $v_{j,y}$ and the source packet $p_{z\cap w}$ is instantly decodable for vertices $v_{k,z}$ and $v_{m,w}$. Thus, from \cref{cor:extended-benefit}, $\mathcal{P}$ will benefit all four vertices. Thus, these vertices are guaranteed to form a proper 4-clique.

Now if we want to assume $K=5$ by adding one extra vertex $v_{n,u}$ pairwise satisfying BC1 with all the above four vertices. We may fall into the following situation:
\begin{itemize}
\item The intersection of any three packets sets of the above 5 vertices is empty (because BC1 guarantees only pairwise intersections between packet sets).
\item Every source packets in the intersection of any two packet sets is in the Lack set of at least one of the three other vertices.
\end{itemize}
In this situation, there does not exist a packet combination that can instantly decode or further aggregate all five vertices. Thus, this 5-clique is not proper. This situation may occur for any $K geq 5$ and the second statement of the theorem follows.
\end{proof}
}

\subsection{Constraining BC2}
It is clear from \thref{th:BC-bound} that BC2 is more critical than BC1. Indeed, when added to BC1, BC2 reduces the bound on the largest $K$ for which the Simple P-VACs can guarantee the formation of proper K-cliques. Actually, it easy to show that, even if one element of the packet set of one vertex is not in the Has set of its adjacent vertex according to BC2, proper cliques cannot be guaranteed for $K>3$. This is illustrated in Example \ref{ex:Improper-clique} in \appref{app:examples}. 
\ignore{
For instance, let us take the example in Fig. \ref{fig:example} depicting four vertices, each pair of which satisfying BC2, and thus forming a 4-clique according to it. Defining $\mathcal{B}(\mathcal{P})$ as the set of vertices benefiting (achieving decodability or aggregation) from $\mathcal{P}$, we can see that there does not exist any coding combination that can benefit all the vertices.

\begin{figure}[t] \label{fig:iImproper-clique}
\centering
\includegraphics[width=1\linewidth]{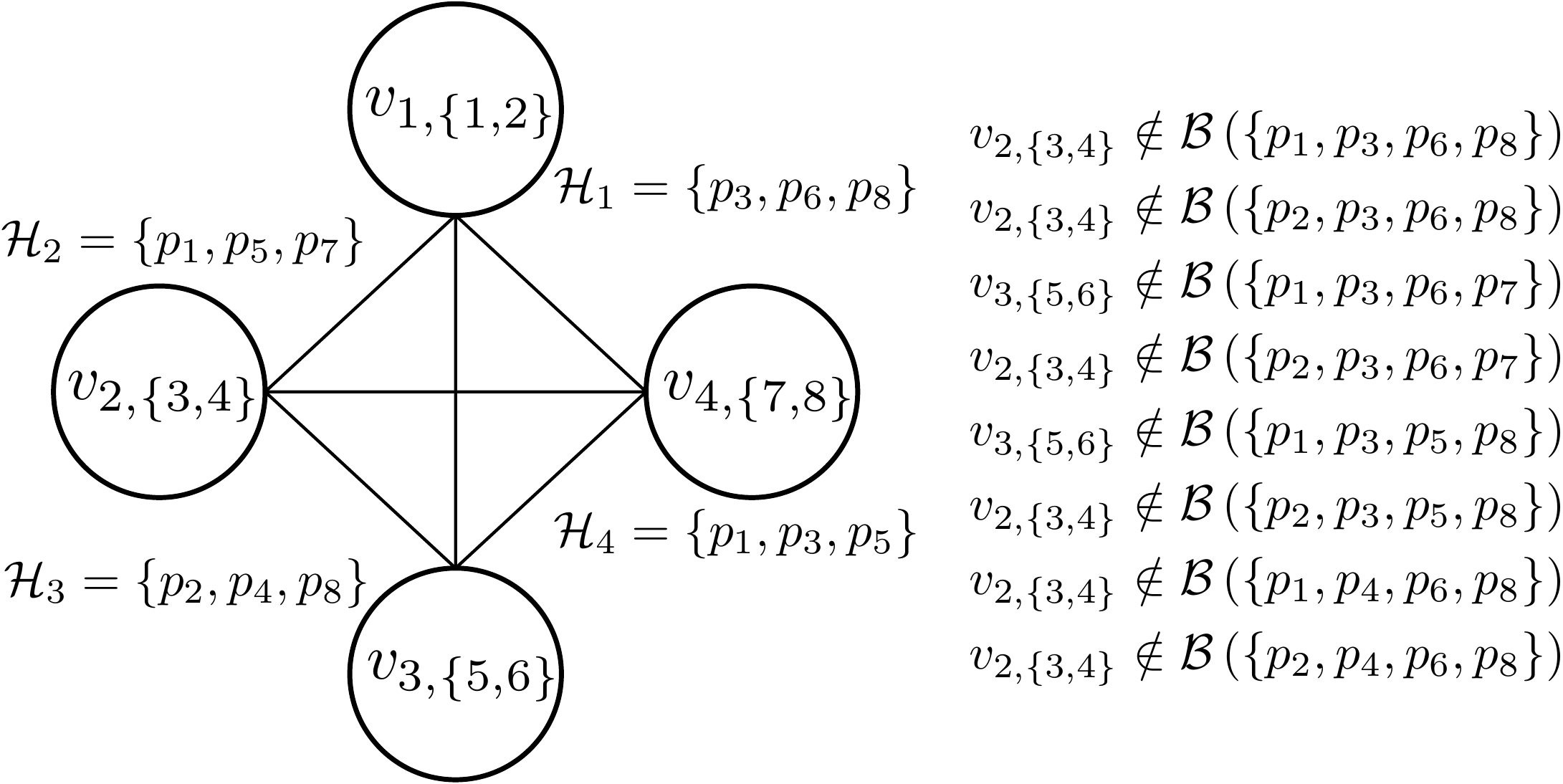}\\
\caption{Example showing four vertices, each pair of which satisfying BC2. A packet combination benefiting each vertex must have a least one source packet from its packet set. For example, a packet combination $\mathcal{P}$ benefiting $v_{1,\{1,2\}}$ must include $p_1$ or $p_2$ or both. Now to serve the other vertices without violating the benefit of $v_{1,\{1,2\}}$, it is easy to see that $\mathcal{P}$ must have at least two elements from $\mathcal{H}_1$. On the right, we can see all combinations of minimum size that satisfy both above conditions. Any larger combination will definitely include one of these combinations in it. We can see that each of these combinations is not benefiting  one other vertex\ignore{ (Other vertices may also not benefit from each of these combinations but only one is shown as this is sufficient to eliminate the combination validity as a solution for all the vertices of the clique)}.}
\label{fig:example}
\end{figure}
}

Consequently, this condition should be restricted to its most by enforcing the following new condition:
\begin{definition}[BC2$^*$: Constrained BC2]
Two vertices $v_{i,x}$ and $v_{k,y}$, such that $x\cap y = \emptyset$, are adjacent only if $x\subseteq \mathcal{H}_k$ and $y\subseteq \mathcal{H}_i$.
\end{definition}
Note that this is the trivial generalization of Condition C2 in the IDNC graph for vertices representing the demand of different source packets.

\subsection{Constraining BC1}
\ignore{Before proceeding with this section, we need to get reminded that this Condition BC1 in the O2-ONC graph should reflect the vertex restrictions that are equivalent to the counterpart restrictions of C1 in the IDNC graph. Consequently, the final result of this constraint must depict the pairwise restrictions of the vertices of each receiver on those of the other receiver in terms of common benefiting packet combinations. In other words, the resulting packet combinations from this constraint should include at most two packets (basic packet as opposed to uncoded source packets in the case of restricted vertices in IDNC) and at least on of them should be in the packet set of each vertex benefiting from this packet.}

In this section, we aim to find a valid constraining of BC1. Similar to BC2, we may think of the trivial generalization of C1 in the IDNC graph by setting two vertices $v_{i,x}$ and $v_{k,y}$ adjacent only if $x=y$. Nonetheless, we will show in the next theorem that we can still construct proper cliques with less restrictive conditions.\ignore{ In what follows, we define $r_{\nsubseteq s}$ as the the elements of set $r$ that are not in set $s$ (i.e. $r_{\nsubseteq s} =  r\setminus s$).}
\ignore{
\begin{theorem}\label{th:DC1}
If the adjacency between any two vertices $v_{i,x}$ and $v_{k,y}$, such that $x\cap y \neq \emptyset$, is set only if the following decoding condition (DC) is satisfied:
\begin{equation}
\mbox{DC1}: x\subseteq y \quad \mbox{OR} \quad y\subseteq x\;,
\end{equation}
then every resulting clique among these vertices will have at least one combination that will instantly decode all them.
\end{theorem}

\begin{proof}
The proof can be found in Appendix \ref{app:DC1*}.
\ignore{
Assume any set of vertices are all adjacent to each other according to DC1. Let $d_{min}$ be the dimension of the one or multiple vertices having the smallest packet set sizes. If more than one vertex have dimension $d_{min}$, then all such vertices will have equal packet sets $x_{min}$ or else they will violate Condition DC1. Moreover, any other vertex with larger dimension will definitely include all the packets of $x_{min}$ in their packet sets or else they will violate Condition DC1. Thus, any source packet in $x_{min}$ or a combination of its packets will be instantly decodable for all the vertices of the clique.
}
\end{proof}
}

\begin{theorem}\label{th:C1*}
If the adjacency between any two vertices $v_{i,x}$ and $v_{k,y}$, such that $x\cap y \neq \emptyset$, is set only if the following benefiting condition is satisfied:
\begin{equation}\label{eq:C1*}
\mbox{BC1}^*:
\begin{cases}
 \quad x \subseteq y    \quad~ & |y| \geq |x| = 1\\
 \quad |x\cap y| \geq |x| - 1    \quad~ &  |y|=|x| \geq 2\\
 \quad |(x\setminus\mathcal{H}_k)\cap y| \geq |x\setminus\mathcal{H}_k| - 1    \quad~ &  |y|>|x| \geq 2,
\end{cases}
\end{equation}
\ignore{Then, there will exist a packet combination that is either instantly decodable or aggregating for all the vertices forming a clique after this condition. In other words,} then all resulting cliques between such vertices are proper.
\end{theorem}

\begin{proof}
The proof can be found in Appendix \ref{app:C1*}.
\end{proof}


\ignore{
\begin{hypothesis}
If the adjacency between any two vertices $v_{i,x}$ and $v_{k,y}$, such that $x\cap y \neq \emptyset$, is set if and only:
\begin{equation}
\mbox{BC1}^*: x_{\nsubseteq \mathcal{H}_k} \subseteq  y \quad \mbox{OR}\quad y_{\nsubseteq \mathcal{H}_i} \subseteq  x
\end{equation}
then there will exist a packet combination that is instantly decodable for all the vertices forming any clique after this condition.
\end{hypothesis}
}

\section{O2-ONC Graph: Algorithmic Procedures}
After defining the vertices and P-VACs of our proposed graph model and developing a good understanding of its properties, we can now describe the procedures that any algorithm should follow in order to optimize any desired metric in O2-ONC.

\subsection{Graph Construction} \label{sec:construction}


Before the start of the delivery process, no receiver has stored O2-IP yet and thus the graph starts as an IDNC graph. This will require $O(MN)$ steps as the maximum possible number of vertices in the graph is equal to $MN$ (when all receivers have no packets at all). Along the delivery process, the vertices of the graph progressively aggregate at each instance (if any) a receiver stores an O2-IP that satisfy \lref{lem:aggregation}. Consequently, there will be no need to use \dref{def:generalized-vertex} to generate the vertices before each transmission.

With the vertices established, each pair are checked to determine whether they should be set adjacent (i.e. whether they satisfy either of the P-VACs  BC1$^*$ or BC2$^*$). Since this check needs to be made for every two vertices without repetition, we need a $\binom{V}{2}$ checks, $V$ being the graph's vertex set size bounded by $O(MN)$. Thus, the overall complexity of the edge generation process is $O(M^2N^2)$.

\subsection{Served Vertex Selection}
Similar to IDNC, the selection of a packet combination is usually done to achieve a specific target, such as minimizing the completion time \cite{TON10-CD}, minimizing the decoding delay \cite{GC10}, minimizing the in-order-delivery delay \cite{Sundararajan2009} or providing a certain metric of quality of service \cite{5072357,6030131}. Usually, these metrics can be represented\ignore{ in one way or another} by assigning a weight $w_{i,x}$ to each vertex $v_{i,x}$ in the graph, which gives a certain priority of service to its receiver or its packet set in order to achieve the target. Thus, the clique selection can be done by solving a maximum weight clique problem on the constructed graph.

It is well known that finding or approximating the maximum weight clique in a graph is NP-hard \cite{\ignore{Garey1979,}Ausiello1999}. However, there exist several exact algorithms that solve this problem in polynomial time for moderate size graphs (\cite{Yamaguchi2008} and references therein). Nonetheless, the complexity of these algorithms may still be prohibitive for some applications \cite{Yamaguchi2008}. In this case, \cite{TON10-CD,GC10} and many other recent works in the IDNC context have developed an $O(M^2N)$ iterative vertex search procedure\ignore{ with updatable weights after each vertex selection} and showed that it can achieve a very small degradation compared to the optimal clique selection. First, define the modified weigh $\omega_{i,x}(\mathcal{G}_s)$ for each vertex $v_{i,x}$ in the sub-graph $\mathcal{G}_s$ as:
\begin{equation}\label{modified-weights}
\omega_{i,x}(\mathcal{G}_s) = w_{i,x}\cdot \sum_{v_{k,y}\in \mathcal{N}_{\mathcal{G}_s}(v_{k,y})} w_{k,y}
\end{equation}
where $\mathcal{N}_{\mathcal{G}_s}(v_{i,x})$ is the set of vertices adjacent to vertex $v_{i,x}$ in sub-graph $\mathcal{G}_s$. Consequently, this modified weight is large for vertices both having large raw weights and adjacent to vertices in $\mathcal{G}_s$ with large raw weights themselves. This last condition helps in selecting the vertices having high chance in being in a maximal clique with a large total sum of raw weights.

The procedure determines the desired clique $\kappa$ iteratively by selecting the vertex with the maximum modified weight in each iteration. After each vertex selection, the modified weights are re-computed in the subgraph of vertices that are adjacent to all previously selected vertices in $\kappa$, and then the vertex having the new maximum modified weight is selected. This procedure stops when no more adjacent vertices to all previously selected vertices in $\kappa$ are found.

\subsection{Determination of Packet Combination}
Once the desired clique $\kappa$ is determined, the packet combination $\mathcal{P}^\star$ is simply determined as follows. First, pick the vertex $v_{i,x_{min}}$ with smallest dimension in $\kappa$. As shown in the proof of \thref{th:C1*}, the packet set $x_{min}$ of this vertex should be all, or all except for one, in the union of the packet set and Has set of every vertex in $\kappa$ intersecting with it. Otherwise, Condition BC1$^\star$ will be violated. The vertices in $\kappa$ with non-intersecting packet sets with $x_{min}$ will have all their packet sets in $\mathcal{H}_i$ and vice versa, as mandated by BC2$^\star$. Thus, the packet combination $\mathcal{P}^\star_1$, including all the source packets in $x_{min}$, will benefit all the vertices in $\kappa$ having intersecting packet sets with $x_{min}$, and will be all in the Has sets of the remaining vertices of $\kappa$. Thus, the vertices having intersecting packet sets with $x_{min}$ are removed from $\kappa$ as they are already served by $\mathcal{P}^\star_1$.

For the remaining vertices in $\kappa$, we can repeat the above procedure several times (thus finding $\mathcal{P}^\star_2$, $\mathcal{P}^\star_3$, $\dots$) until no vertex remains in $\kappa$. Thus, the combination $\mathcal{P}^\star = \mathcal{P}^\star_1 \cup \mathcal{P}^\star_2 \cup \mathcal{P}^\star_3 \cup \dots$ is definitely a combination that benefits all the vertices of $\kappa$. Note that this procedure is of $O(M)$ complexity since there exists at most $M$ vertices in any given clique.

\ignore{
\subsection{Possible Added Packet}$\quad$\\
According to Corollary \ref{cor:extended-benefit}, sending only a source packets benefits less or the same set receivers than the same source packet combined with any other source packet. Thus, if the maximum weight clique procedure ends up with only one packet $p$ in the coding combination, another packet $q$ can be XORed with $p$ so as to serve more receivers. The selection of packet $q$ can be done based on the same criteria used to select the packet $p$ in the first place to optimize the required parameter. For example, if the parameter to optimize is the completion time, the packet $q$ can be identified as the packet that is common in the Wants sets of the worst receivers (receivers with largest Wants sets and erasure probabilities).
}


\subsection{Overall Complexity}
From the previous descriptions, we can see that the maximum complexity of any one procedure is $O(M^2N^2)$, which means that this complexity dominates all others and is the actual complexity of the whole algorithm. Note that this algorithm is more complex than the corresponding algorithms in IDNC with only a factor of $O(N)$. Most importantly, this complexity is much lower than that of handling exponentially increasing virtual queues (in $M$ and possibly $N$ for some cases), which is a common practice in most works allowing storage of undecoded packet combinations in ONC \cite{4476183,4594999,Drinea2009,Gatzianas2010,Wang2010}.

\section{Case Study: Completion Time Problem}
The completion time problem is one of the most fundamental problems in packet delivery\ignore{ processes} using network coding. The completion time is defined as the total number of transmissions until all the receivers get all their missing packets. In this section, we present the performance of our proposed algorithmic procedures using our designed O2-ONC graph in solving the completion time problem\ignore{ in this subclass of ONC} over heterogenous erasure channels. We also compare its performance to both IDNC and the global optimal performance achievable by any linear network code. IDNC and O2-ONC are tested for both the optimal and search clique approaches. In our proposed O2-ONC graph, the raw weights of vertices are set as in \cite{TON10-CD}, which were found to be one of the best weights to reduce the completion time using IDNC. For fairness of comparison, all compared approaches are assumed to use linearly independent coding combinations across all transmissions.

Fig. \ref{fig:case-study} depicts the case-study comparison against the number of receivers (for 30 packets and average erasure probability of 0.15) and the average erasure probability (for 30 packets and 60 receivers). Both comparisons show the expected outperformance of our proposed algorithms compared to IDNC for larger number of receivers and erasure probabilities, especially for the more practical clique search approach. They also show that even the simple clique search O2-ONC approach almost achieves optimality over all linear network codes.

\begin{figure}[t]
\centering
\includegraphics[width=1\linewidth]{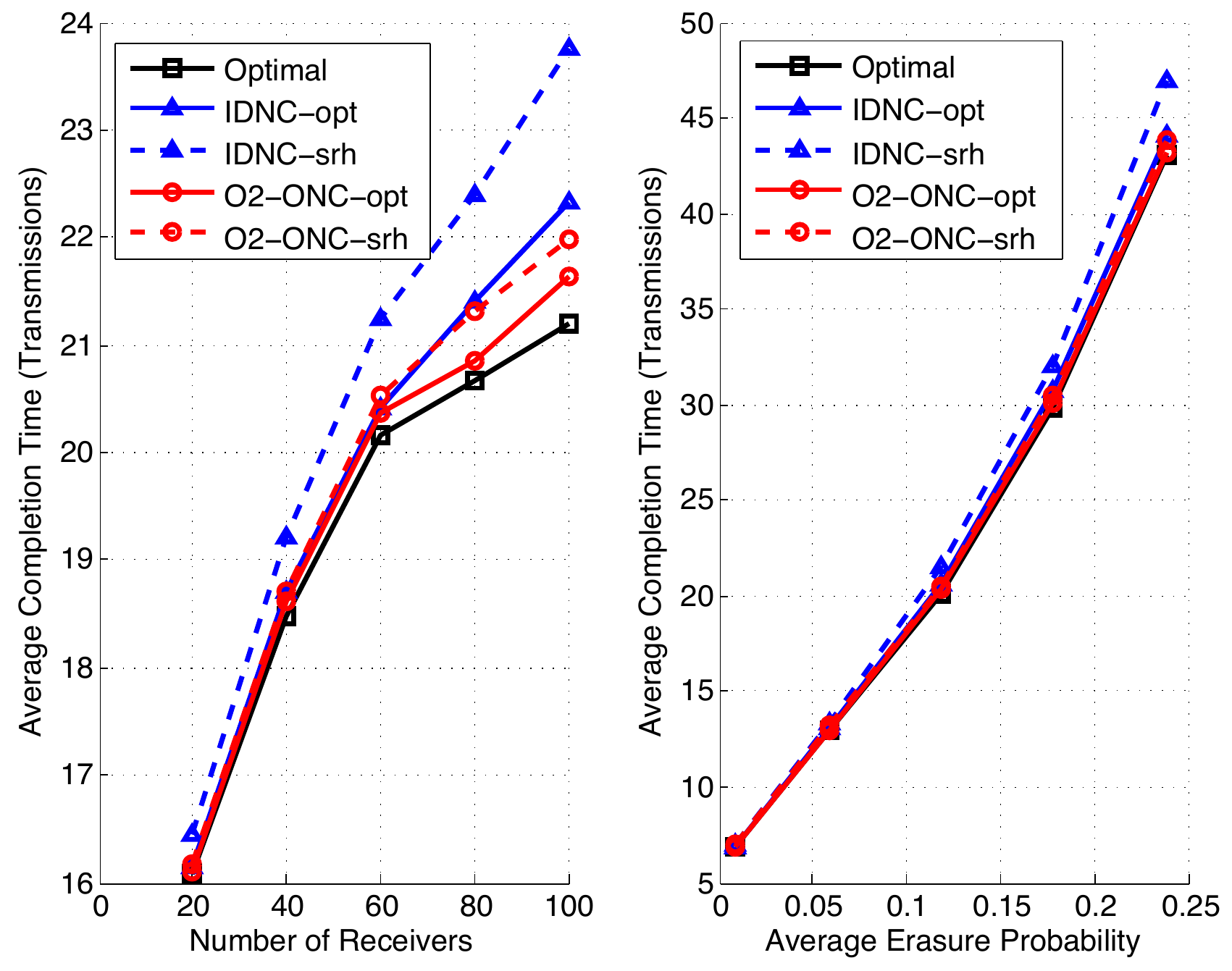}\\
\caption{Completion time comparisons against IDNC and global optimal solution.}
\label{fig:case-study}
\end{figure}

\section{Conclusion}
This paper introduced a graph model and algorithmic procedure that can optimize any metric in O2-ONC, with only a linear increase in complexity (with the number of packets $N$) compared to the well-studied IDNC solutions. This graph model is based on a new generalized definition of the graph vertices and the novel notion of vertex aggregation. Constrained P-VACs were derived and proven to guarantee the benefit of every clique in the graph by only one coded transmission. Simulation results for solving the completion time problem using our proposed framework has shown closer performance to the optimal one, compared to the IDNC solution, especially for large and harsh-channel networks.

\bibliographystyle{IEEEtran}
\bibliography{IEEEabrv,bibfile}

\appendices

\section{Clarifying Examples} \label{app:examples}
\begin{example}\label{ex:generalized-vertex}
A receiver $r_i$ with $\mathcal{H}_i = \{p_1,p_3,p_5,p_7\}$ and $\mathcal{W}_i = \{p_2,p_4,p_6,p_8\}$ is storing the following combinations:
\begin{itemize}
\item $\mathcal{P}_a = a_1 p_1 + a_2 p_2 + a_3 p_4$.
\item $\mathcal{P}_b = b_1 p_3 + b_2 p_4 + b_3 p_7 + b_4 p_8$.
\end{itemize}
Clearly, these 2 O2-IPs collectively include only 3 source packets from $\mathcal{W}_i$ and satisfy the conditions in \dref{def:generalized-vertex}. Consequently, these 3 source packets can be represented in the O2-ONC graph by only ONE vertex $v_{i,\{2,4,8\}}$. This representation means that $r_i$ can decode all three source packets if it receives only ONE of these source packets or ONE appropriate combination of any of  them that satisfy the above decodability conditions. For example, if $r_i$ receives $p_2$, $p_4$, $p_8$, or any combination of two or three or them (and possibly combined with other packets in $\mathcal{H}_i$, it will have 3 equations in 3 unknown source packets, which will result in their decoding.
\end{example}

\begin{example} \label{ex:aggregating-packets}
In the same scenario of Example \ref{ex:generalized-vertex}, if $r_i$ receives a coded packet $\mathcal{P}_c = c_1 p_2 + c_2 p_5 + c_3 p_6 + c_4 p_8$, it can be seen that, by excluding the Has set packets of $r_i$ from $\mathcal{P}_c$, the remaining source packets belong to either vertex $v_{i,\{2,4,8\}}$ or $v_{i,6}$. Despite the fact that $\mathcal{P}_c$ has more than 2 (in fact 3) packets from $\mathcal{W}_i$, it still is considered as an $O2-IP$ according to \dref{def:O2-IP} and its reception at $r_i$ results in aggregating the vertices $v_{i,\{2,4,8\}}$ and $v_{i,6}$ into one vertex $v_{i,\{2,4,6,8\}}$. Indeed, $r_i$ now needs to receive only ONE of these source packets or ONE combination of them (again possibly with other packets from $\mathcal{H}_i$), which can be used along with $\mathcal{P}_a$, $\mathcal{P}_b$ and $\mathcal{P}_c$ to decode all of these four source packets. Thus, they must be all represented by ONE vertex in the O2-ONC graph.
\end{example}

\begin{example}\label{ex:proper-clique}
\begin{figure}[t]
\centering
\includegraphics[width=0.6\linewidth]{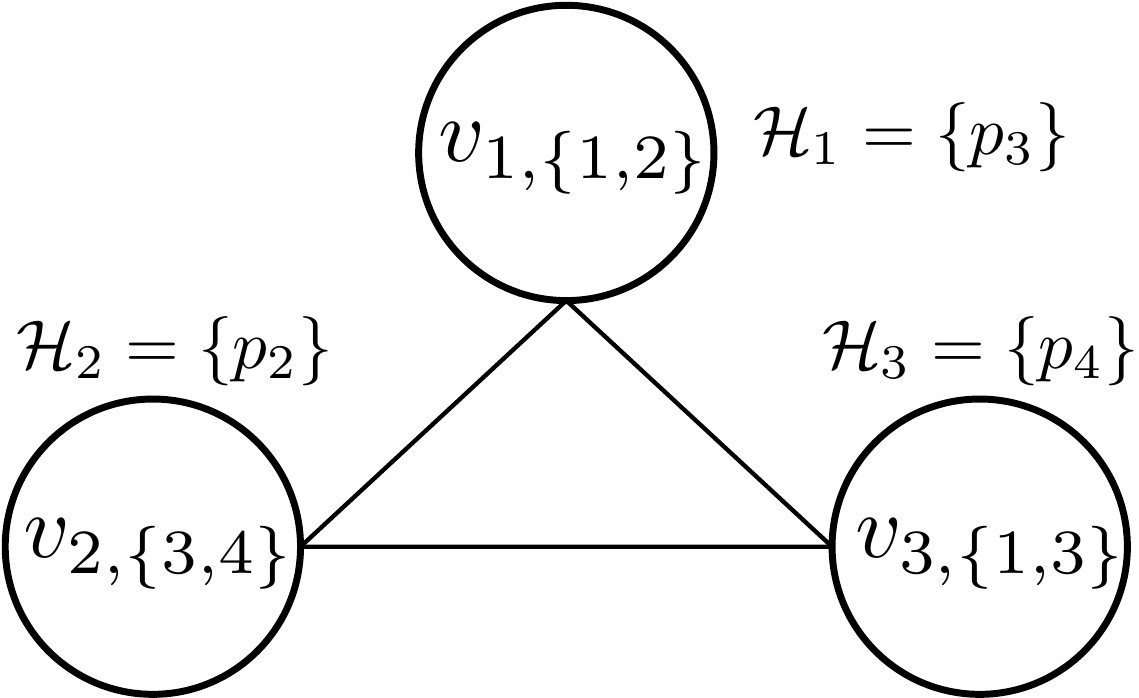}\\
\caption{Example of a proper clique between 3 vertices in the O2-ONC graph (The rest of the vertices are hidden to emphasize this clique). Each two vertices satisfy either of the Simple P-VACs in \dref{def:Simple-PVACs}. \ignore{ A coding combination $\{p_1,p_3\}$ or $\{p_1,p_4\}$ or $\{p_2,p_3\}$ or $\{p_2,p_4\}$ will definitely either decode each of these vertices or will aggregate it with another vertex of the same receiver. For example, $\{p_1,p_3\}$ will decode vertices $v_{1,\{1,2\}}$ and  $v_{3,\{1,3\}}$, and will aggregate $v_{2,\{3,4\}}$ and $v_{2,1}$ into one vertex $v_{2,\{1,3,4\}}$.}}
\label{fig:proper-clique}
\end{figure}
Fig. \ref{fig:proper-clique} shows an example of a proper 3-clique in the shown graph scenario. We can easily see that all combinations $\{p_1,p_3\}$ or $\{p_1,p_4\}$ or $\{p_2,p_3\}$ or $\{p_2,p_4\}$ will definitely either decode each of these vertices or will aggregate it with another vertex of the same receiver. For example, $\{p_1,p_3\}$ will decode vertices $v_{1,\{1,2\}}$ and  $v_{3,\{1,3\}}$, and will aggregate $v_{2,\{3,4\}}$ and $v_{2,1}$ into one vertex $v_{2,\{1,3,4\}}$.
\end{example}

\begin{example} \label{ex:Improper-clique}
Fig. \ref{fig:Improper-clique} depicts four vertices, each pair of which satisfying BC2, and thus forming a 4-clique according to it. Define $\mathcal{B}(\mathcal{P})$ as the set of vertices benefiting (achieving decodability or aggregation) from $\mathcal{P}$. A packet combination benefiting each vertex must have a least one source packet from its packet set. For example, a packet combination $\mathcal{P}$ benefiting $v_{1,\{1,2\}}$ must include $p_1$ or $p_2$ or both. Now to serve the other vertices without violating the benefit of $v_{1,\{1,2\}}$, it is easy to see that $\mathcal{P}$ must have at least two elements from $\mathcal{H}_1$. On the right, we can see all combinations of minimum size that satisfy both above conditions. Any larger combination will definitely include one of these combinations in it. We can see that each of these combinations is not benefiting one other vertex, and thus there does not exist any coding combination that can benefit all the vertices.\ignore{ (Other vertices may also not benefit from each of these combinations but only one is shown as this is sufficient to eliminate the combination validity as a solution for all the vertices of the clique).}
\begin{figure}[t] 
\centering
\includegraphics[width=1\linewidth]{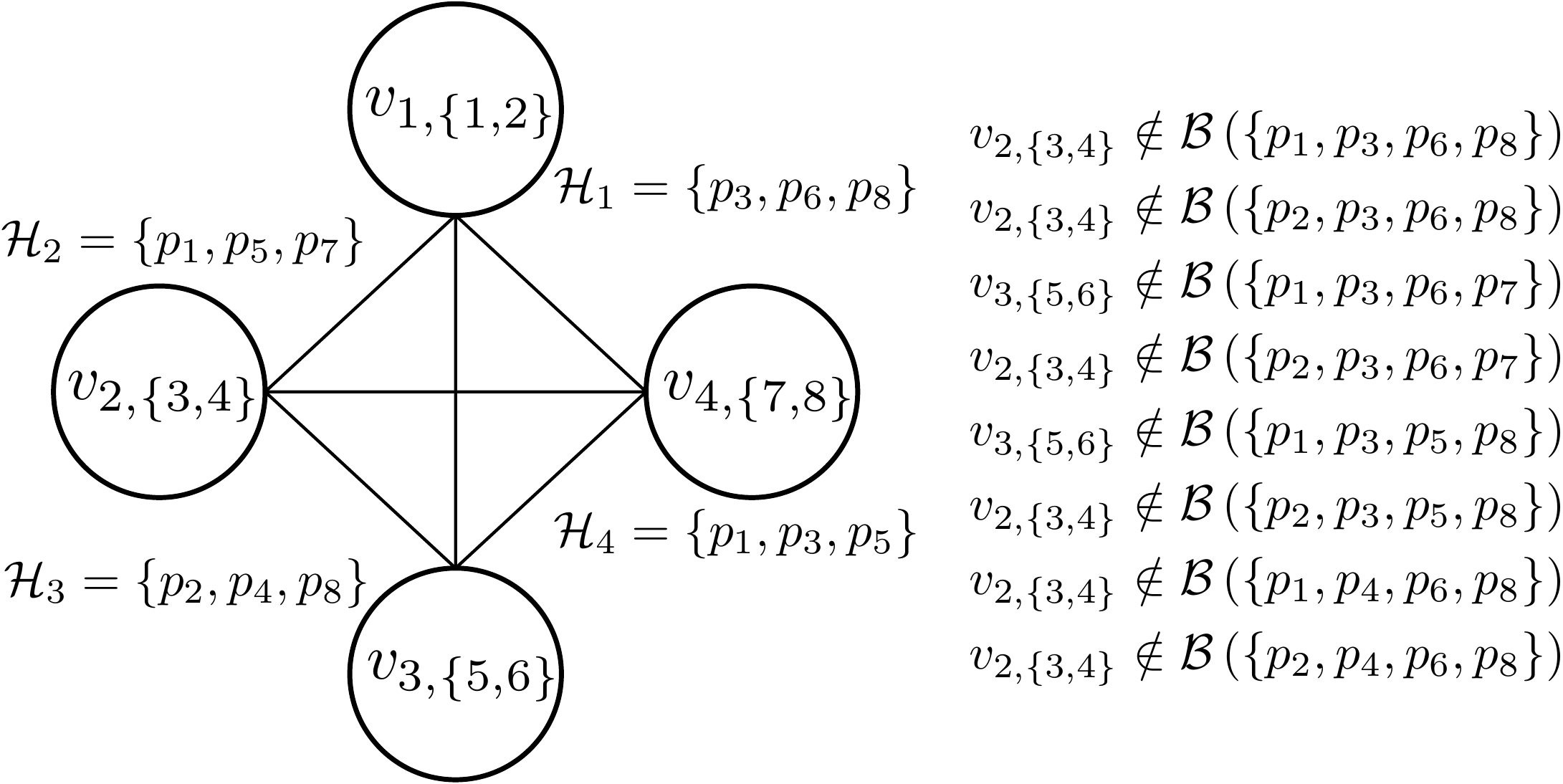}\\
\caption{Example showing four vertices, each pair of which satisfying BC2. \ignore{A packet combination benefiting each vertex must have a least one source packet from its packet set. For example, a packet combination $\mathcal{P}$ benefiting $v_{1,\{1,2\}}$ must include $p_1$ or $p_2$ or both. Now to serve the other vertices without violating the benefit of $v_{1,\{1,2\}}$, it is easy to see that $\mathcal{P}$ must have at least two elements from $\mathcal{H}_1$. On the right, we can see all combinations of minimum size that satisfy both above conditions. Any larger combination will definitely include one of these combinations in it. We can see that each of these combinations is not benefiting  one other vertex\ignore{ (Other vertices may also not benefit from each of these combinations but only one is shown as this is sufficient to eliminate the combination validity as a solution for all the vertices of the clique)}.}}
\label{fig:Improper-clique}
\end{figure}
\end{example}

\section{Proof of Lemma \ref{lem:aggregation}} \label{app:aggregation}
We can deal with the $|x|-1$ ($|y|-1$) O2-IP packets of vertex $v_{i,x}$ ($v_{i,y}$) as a set of $|x|-1$ ($|y|-1$) independent linear equations in $|x|$ ($|y|$) variables. When the packet $\mathcal{P}$, defined as shown above, is received by $r_i$, the packets belonging to $\mathcal{H}_i$ (if any) in $\mathcal{P}$ can be cancelled and we get one extra equation in variables that belong to $x\cup y$ only. But since $\mathcal{P}\cap x \neq \emptyset$ and $\mathcal{P}\cap y \neq \emptyset$, then this extra equation, along the $|x|-1$ equations of $v_{i,x}$ and $|y|-1$ equations of $v_{i,y}$ represent $|x|-1 + |y|-1 + 1 = |x|+|y|-1$ equations all in the $|x|+ |y|$ variables of $x\cup y$. Consequently, we can solve for all these variables  (i.e. decode all these packets) by only receiving one extra equation in the same variables. Thus, as per \dref{def:generalized-vertex}, all these packets must be represented by only one vertex $v_{i,x\cup y}$.

\section{Proof of Theorem \ref{th:BC-bound}} \label{app:BC-bound}
In this proof, we denote by $p_{\{s\}}$ any source packets in set $s$.

We first prove the first statement of the theorem. For the case of $K = 2$:
\begin{itemize}
\item If BC1 is true for any two vertices $v_{i,x}$ and $v_{k,y}$, then $\mathcal{P} = p_{\{x\cap y\}}$ will be instantly decodable for both receivers $r_i$ and $r_k$.
\item If BC2 is true for any two two vertices $v_{i,x}$ and $v_{i,x}$, then $\mathcal{P} = \left\{p_{\{x\cap\mathcal{H}_k\}}, p_{\{y\cap\mathcal{H}_i\}}\right\}$ will be instantly decodable for both receivers $i$ and $k$.
\end{itemize}
For the case of $K = 3$:
\begin{itemize}
\item If BC1 is true for only two vertices $v_{i,x}$ and $v_{k,y}$, then from our analysis of the two receivers case, there exists an instantly decodable packet $ p_{\{x\cap y\}}$ for these two vertices. Now, if another third vertex $v_{m,z}$ of receiver $r_m$ is adjacent to both $v_{i,x}$ and $v_{k,y}$ (according to either BC1 or BC2), then the combination $\mathcal{P} = \left\{p_{\{x\cap y\}}, p_{\{z\}}\right\}$ will definitely be either instantly decodable or aggregating (from \lref{lem:aggregation}) for $r_i$ ($r_k$) if $p_{\{z\}} \in x\cup\mathcal{H}_i$ ($p_{\{z\}} \in y\cup\mathcal{H}_k$) or $p_{\{z\}}\in \mathcal{W}_i\setminus x$ ($p_{\{z\}}\in\mathcal{W}_k\setminus y$), respectively. The combination $\mathcal{P}$ will also be instantly decodable or aggregating to vertex $r_m$ if $p_{\{x\cap y\}} \in z\cup\mathcal{H}_m$ or $p_{\{x\cap y\}}\in \mathcal{W}_m\setminus z$, respectively. Thus, the three vertices will definitely benefit from this combination.
\item If BC2 is true pairwise between three vertices $v_{i,x}$, $v_{k,y}$ and $v_{m,z}$, consider the following combination:
    \begin{equation}
    \mathcal{P} = \left\{p_{\{x\cap \mathcal{H}_k\}}, p_{\{y\cap \mathcal{H}_m\}}, p_{\{z\cap\mathcal{H}_i\}}\right\}
    \end{equation}
    From $r_i$'s perspective, it can eliminate $p_{\{z\cap\mathcal{H}_i\}}$ (because it has it). The remaining combination $\left\{p_{\{x\cap \mathcal{H}_k\}}, p_{\{y\cap \mathcal{H}_m\}}\right\}$\ignore{ If $p_{y\cap \mathcal{H}_m} \in \mathcal{H}_i$, then this combination will be instantly decodable for vertex $v_{i,x}$. If not, than there exists another vertex $v_{i,w}$ such that $p_{y\cap \mathcal{H}_m}\in t$. Thus from \lref{lem:aggregation}, this will results in the aggregation of $v_{i,x}$ and $v_{i,w}$ into one vertex $v_{i,x\cup w}$.} will definitely be either instantly decodable or aggregating (from \lref{lem:aggregation}) for $r_i$ if $p_{\{y\cap \mathcal{H}_m\}} \in x\cup\mathcal{H}_i$ or $p_{\{y\cap \mathcal{H}_m\}}\in \mathcal{W}_i\setminus x$, respectively. Similar arguments can be proven for receivers $r_k$ and $r_m$ after eliminating their Has packets $p_{\{x\cap \mathcal{H}_k\}}$ and $p_{\{y\cap \mathcal{H}_m\}}$, respectively.
\end{itemize}
$\quad$\\
Now to prove the second statement, let us first consider the case of $K=4$. Assume four vertices $v_{i,x}$, $v_{j,y}$, $v_{k,z}$ and $v_{m,w}$ forming a 4-clique because they pairwise satisfy Condition BC1. For this clique, we can always generate the following packet combination:
\begin{equation}
\mathcal{P} = \left\{p_{\{x\cap y\}}, p_{\{z\cap w\}}\right\}
\end{equation}
Clearly, the first source packet $p_{\{x\cap y\}}$ is instantly decodable for vertices $v_{i,x}$ and $v_{j,y}$ and the second source packet $p_{\{z\cap w\}}$ is instantly decodable for vertices $v_{k,z}$ and $v_{m,w}$. Thus, $\mathcal{P}$ will definitely be either instantly decodable or aggregating (from \lref{lem:aggregation}) for $r_i$ ($r_j$) if $p_{\{z\cap w\}} \in x\cup\mathcal{H}_i$ ($p_{\{z\cap w\}} \in y\cup\mathcal{H}_j$) or $p_{\{z\cap w\}}\in \mathcal{W}_i\setminus x$ ($p_{\{z\cap w\}}\in\mathcal{W}_j\setminus y$), respectively. Similarly, $\mathcal{P}$ will definitely be either instantly decodable or aggregating (from \lref{lem:aggregation}) for $r_k$ ($r_m$) if $p_{\{x\cap y\}} \in z\cup\mathcal{H}_k$ ($p_{\{x\cap y\}} \in w\cup\mathcal{H}_m$) or $p_{\{x\cap y\}}\in \mathcal{W}_k\setminus z$ ($p_{\{x\cap y\}}\in\mathcal{W}_m\setminus w$), respectively.  Thus, these vertices are guaranteed to form a proper 4-clique.

Now assume $K=5$ by adding one extra vertex $v_{n,u}$ pairwise satisfying BC1 with all the above four vertices. We may fall into the following situation:
\begin{itemize}
\item The intersection of any three packets sets of the above 5 vertices is empty (because BC1 guarantees only pairwise intersections between packet sets).
\item Every source packet in the intersection of any two packet sets is in the Lack set of at least one of the three other vertices.
\end{itemize}
In this situation, there does not exist a packet combination that can instantly decode or further aggregate all five vertices. Thus, this 5-clique is not proper. This situation may occur for any $K \geq 5$ and the second statement of the theorem follows.

\ignore{
\section{Proof of Theorem \ref{th:DC1}} \label{app:DC1*}
Assume any set of vertices are all adjacent to each other according to DC1. Let $d_{min}$ be the dimension of the one or multiple vertices having the smallest packet set sizes. If more than one vertex have dimension $d_{min}$, then all such vertices will have equal packet sets $x_{min}$ or else they will violate Condition DC1. Moreover, any other vertex with larger dimension will definitely include all the packets of $x_{min}$ in their packet sets or else they will violate Condition DC1. Thus, any source packet in $x_{min}$ or a combination of its packets will be instantly decodable for all the vertices of the clique.
}

\section{Proof of Theorem \ref{th:C1*}}\label{app:C1*}

Assume any set of vertices that are all adjacent to each other according to their pairwise satisfaction of BC1$^*$. Let $d_{min}$ be the dimension of the one or multiple vertices having the smallest dimension in this clique and let $v_{i,x_{min}}$ be one of such vertices.\\
$\quad$\\
\textbf{Case 1}: $d_{min} = 1$ \\
In this case, there exists at least one vertex $v_{i,x_{min}}$ with dimension 1, thus getting all its adjacency to all the other vertices of the clique by the first entry of \eqref{eq:C1*} when $\min\{|x|,|y|\}=1$. \ignore{Thus, the adjacency of this vertex to all other vertices of the clique becomes a special instance of \thref{th:DC1}.} Consequently, all the vertices of the clique will have $x_{min}$ (consisting of one packet) in their packet sets and thus this packet will be instantly decodable for all the vertices of this clique.\\
$\quad$\\
\textbf{Case 2}: $|y| \geq |x| \geq 2$,  $x\cap\mathcal{H}_k = \emptyset$ \\
Note that this case covers the second and third entries of \eqref{eq:C1*} when $x\setminus\mathcal{H}_k = x$.

In this case, any other vertex than $v_{i,x_{min}}$ must have at least $d_{min}-1$ source packets from $x_{min}$ in their packet sets, or else BC1$^*$ will be violated. The number of different sub-combinations of these $d_{min} -1$ source packets out to $d_min$ ones is equal to $\binom{d_{min}}{d_{min}-1} = d_{min}$ sub-combinations. Out of these $d_{min}$ possible sub-combinations, the number of sub-combinations having one arbitrary common packet $p\in x_{min}$ is equal to $\binom{d_{min}-1}{d_{min}-2} = d_{min} - 1$. Thus, there exists only one sub-combination $y =  x_{min}\setminus p$, ($|y| = d_{min}-1$) that does not have this one common packet $p$.

Now as stated above, every vertex other than $v_{i,x_{min}}$ in the formed clique will have one of the $d_{min}$ sub-combinations of size $d_{min}-1$ in its packet set. Note that if there exists more than $d_{min}$ vertices in the clique, then from the pigeonhole principle, more than one vertex will have the same combination in its packet set. Consequently, each of these vertices will either include the common packet $p$ or a the whole sub-combination $y$. Thus, the packet combination $\mathcal{P} = \{p, p'\}$ ($p'$ being any source packet in $y$) will have the following effect on all the vertices of the clique:
\begin{itemize}
\item For any vertex $v_{k,z}$ having $p \in z$, $\mathcal{P}$ will be either instantly decodable or aggregating if $p' \in\mathcal{H}_k$ or $p' \in\mathcal{W}_k$, respectively.
\item For any vertex $v_{m,w}$ having $p' \in w$, $\mathcal{P}$ will be either instantly decodable or aggregating if $p \in\mathcal{H}_m$ or $p \in\mathcal{W}_m$, respectively.
\end{itemize}
Clearly, $\mathcal{P}$ will be also instantly decodable for $v_{i,x_{min}}$. Thus, $\mathcal{P}$ is instantly decodable or aggregating to all the vertices of the clique, thus making the clique proper.\\
$\quad$\\
\textbf{Case 3}: $|y| > |x| \geq 2$, $x\cap\mathcal{H}_k \neq \emptyset$\\
Again in this case, every vertex $v_{k,y}$ other than $v_{i,x_{min}}$ in the formed clique will have one of the $d_{min}$ sub-combinations of size $d_{min}-1$, but this time these $d_{min}-1$ source packets are distributed between the Has set $\mathcal{H}_k$ of this vertex and its packet set $y$. Thus, the packet combination $\mathcal{P} = x_{min}$ (i.e. the combination of all the source packets in $x_{min}$) will have the following effect on the vertices of the clique:
\begin{itemize}
\item For any vertex $v_{k,y}$ having $|(x_{min}\setminus\mathcal{H}_k)\cap y| = |x_{min}\setminus\mathcal{H}_k|$ (which is equivalent to $x_{min}\setminus\mathcal{H}_k \subseteq y$), all the elements of $\mathcal{P}$ will be either in $\mathcal{H}_k$ and $y$. Thus, the packet combination will be instantly decodable for $y$.
\item For any vertex $v_{k,y}$ having $|(x_{min}\setminus\mathcal{H}_k)\cap y| = |x_{min}\setminus\mathcal{H}_k|-1$ (which is equivalent to $|(x_{min}\setminus\mathcal{H}_k) \setminus y| = 1$), $d_{min}-1$ elements of $\mathcal{P}$ will be either in $\mathcal{H}_k$ and $y$, leaving only one source packet in $\mathcal{P}$ outside these two sets. Thus, $\mathcal{P}$ will aggregate $v_{k,y}$ or another vertex of $r_k$.
\end{itemize}
Clearly, $\mathcal{P}$ will be also instantly decodable for $v_{i,x_{min}}$. Thus, $\mathcal{P}$ is instantly decodable or aggregating to all the vertices of the clique, thus making the clique proper.

\end{document}